\documentclass{article}

\usepackage[margin=1.125in]{geometry}
\usepackage{xcolor,colortbl}
\usepackage{caption}
\usepackage{url}
\usepackage{amssymb,amsmath,amsthm,amsfonts,enumitem,array}
\usepackage{fullpage,nicefrac,comment}
\usepackage{tikz}
\usetikzlibrary{arrows,decorations.pathmorphing,decorations.shapes,snakes,patterns,positioning}
\usepackage[ruled,linesnumbered]{algorithm2e}

\newcommand{\cC}{\ensuremath{\mathcal{C}}}
\newcommand{\cF}{\ensuremath{\mathcal{F}}}

\newcommand{\F}{{\mathbb F}}

\newcommand{\bX}{{\mathbf{X}}}

\newcommand{\cmp}[2]{\begin{minipage}{#1cm}\begin{center}#2\end{center}\end{minipage}}

\newcommand{\inabs}[1]{\left|#1\right|}

\newcommand{\inset}[1]{\left\{#1\right\}}

\newcommand{\inparen}[1]{\left(#1\right)}

\newcommand{\suchthat}{\,:\,}

\renewcommand{\subset}{\subseteq}

\newcommand{\spn}{\ensuremath{\operatorname{span}}}

\newcommand{\tr}{\mathrm{tr}}

\newcommand{\ind}[1]{\ensuremath{\mathbf{1}_{#1}}}

\newcommand{\eps}{\varepsilon}
\renewcommand{\epsilon}{\varepsilon}

\newcommand{\TODO}[1]{\textcolor{red}{\textbf{TODO: #1}}}

\newtheorem{theorem}{Theorem} 
 
\newtheorem{definition}{Definition} 
 
\newtheorem{obs}[theorem]{Observation} 
 
\newtheorem{cor}[theorem]{Corollary} 

\newtheorem{remark}{Remark}

\newtheorem{proposition}[theorem]{Proposition}

\newtheorem*{theorem*}{Theorem}  

\numberwithin{equation}{section}

\title{Repairing Reed-Solomon Codes}
\author{Venkatesan Guruswami\thanks{Carnegie Mellon University, Computer Science Department. Research funded in part by NSF grant CCF-1422045. \texttt{guruswami@cmu.edu}} \ \ and \ \ Mary Wootters\thanks{Carnegie Mellon University, Computer Science Department.  Research funded in part by NSF MSPRF grant
DMS-1400558.  \texttt{marykw@cs.cmu.edu}}}

\begin{document}
\maketitle
\begin{abstract}
We study the performance of Reed-Solomon (RS) codes for the \em exact repair problem \em in distributed storage.
Our main result is that, in some parameter regimes, Reed-Solomon codes are optimal regenerating codes, among MDS codes with linear repair schemes.  Moreover, we give a characterization of MDS codes with linear repair schemes which holds in any parameter regime, and which can be used to give non-trivial repair schemes for RS codes in other settings.

More precisely, we show that for $k$-dimensional RS codes whose evaluation points are a finite field of size $n$, there are exact repair schemes with bandwidth $(n-1)\log((n-1)/(n-k))$ bits, and that this is optimal for any MDS code with a linear repair scheme.  In contrast, the naive (commonly implemented) repair algorithm for this RS code has bandwidth $k\log(n)$ bits.
When the entire field is used as evaluation points, the number of nodes $n$ is much larger than the number of bits per node (which is $O(\log(n))$), and so
this result holds only when the degree of sub-packetization is small.  
However, our method applies in any parameter regime, and to illustrate this for high levels of sub-packetization we give an improved repair scheme for a specific (14,10)-RS code used in the Facebook Hadoop Analytics cluster.

\end{abstract}

\section{Introduction}\label{sec:intro}
This paper studies a polynomial interpolation problem which arises from the use 
of Reed-Solomon codes in distributed storage systems.
In such systems, a large file is encoded and distributed over
many \em nodes. \em  When a node fails, we would like to be able to set up a
 replacement node efficiently using information from the remaining
functional nodes.  The problem of recovering the failed node exactly is known
as the \em exact repair problem. \em 

One traditional solution to the exact repair problem has been to use \em maximum-distance-separating \em (MDS) codes, and in particular \em Reed-Solomon \em (RS) codes.
The RS solution goes as follows.  The original file is broken up into $k$ blocks, and each block is viewed as an element of a finite field $F$.
We interpret the file as a degree $k-1$ polynomial $f$ over $F$: each block is a coefficient of the polynomial.  To distribute the file over the nodes, we
choose $n$ points $\alpha_1,\ldots, \alpha_n \in F$, and send $f(\alpha_i)$ to node $i$.  
Now, if a node fails, we may recover it by looking at the information on any $k$ remaining nodes.  This follows because any $k$ evaluations $f(\alpha_i)$ of a degree $k-1$ polynomial determine the entire polynomial, and hence the contents of the failed node.

This is a non-trivial solution to the exact repair problem, but it's not ideal.  Unfortunately for Reed-Solomon codes, $k$ nodes are also \em necessary \em in this framework, even if all we want to recover is the single failed node.
This is wasteful: we have to read $k$ symbols from $F$ when we only want one. 

Recently, a new approach has emerged, using \em regenerating codes. \em  In this framework, we still use an MDS code to encode the file onto $n$ nodes.  However, a replacement node may choose to download only \em part \em of the contents of each surviving node, rather than being required to download the entire node.  That is, we break up the symbols from $F$ into $t$ sub-symbols in some smaller field $B$ (for example, $B = \inset{0,1}$).
The node is allowed to do some local computation and return one or more sub-symbols, and our goal is to download as few sub-symbols as possible. It turns out that one can do significantly better by downloading fewer sub-symbols from more nodes with than the traditional solution of downloading all of the sub-symbols from each of $k$ nodes.  The number of sub-symbols downloaded in the worst case is called the (exact) \em repair bandwidth \em of the code, over $B$.
The exact repair problem and regenerating codes were first introduced in~\cite{DGWW10}, and have seen a great deal of work since then.  See~\cite{survey} for an excellent survey, and the University of Texas distributed storage wiki~\cite{wiki} for more up-to-date information and references.  

Reed-Solomon codes are a bit maligned in the regenerating codes literature.  A typical paper on regenerating codes---including this one---will mention within the first few paragraphs why the traditional approach with Reed-Solomon codes is not a good idea for the exact repair problem.  Indeed, by now we know of several (non-RS) MDS codes which outperform the traditional RS approach. 
Nonetheless, RS codes are still used!  Because of their ubiquity, it is important to understand what can and cannot be done with Reed-Solomon codes; just because the traditional RS approach isn't a good idea, that does not mean that there isn't some better way to use RS codes.  This was asked as an open question in~\cite{survey}, and is the subject of this paper.

\paragraph{Our contributions.}
We study the exact repair problem for Reed-Solomon codes, and show that one can do much better than the naive scheme.
In fact, we show that 
high (constant) rate Reed-Solomon codes with evaluation points in the whole field are optimal among linear schemes.  Moreover, to the best of our knowledge, they
can significantly outperform all existing constructions in the same parameter regime.
Subsequent work by Ye and Barg~\cite{YB16} has used our
framework to give constructions of RS codes which are optimal in other parameter regimes as well (those with high sub-packetization).

More precisely, our contributions are as follows.
\begin{enumerate}
	\item[(1)] \textbf{Exact repair schemes for high-rate Reed-Solomon codes.}  We show that the repair bandwidth for a rate-$(1 - \eps)$ Reed-Solomon code over a field $F$ and length $n = |F|$ is at most $(n-1)$, over a subfield $B$ of size $1/\eps$, for infinitely many $(n,\eps)$.  In particular, this implies that we can solve the exact repair problem over $GF(2)$ for high-rate Reed-Solomon codes with repair bandwidth $O(n)$ \em bits. \em  Previous constructions of MDS codes with a similar degree of sub-packetization (that is, the number of bits per symbol) require bandwidth $\Omega(n\log(n))$ bits.
	\item[(2)] \textbf{A matching lower bound.} It is easy to see that $k = (1 - \eps)n$ is a lower bound for the repair bandwidth for any MDS code in this setting, and thus our scheme is optimal up to constant factors.  However, we can prove an even stronger lower bound for linear schemes.  We show that our scheme in (1) is optimal for linear repair schemes for MDS codes, even up to the leading constant.  

	\item[(3)] \textbf{A characterization of linear repair schemes, with examples.} We give a characterization for linear exact repair schemes of MDS codes, in terms of the dual code.  For Reed-Solomon codes, (whose dual is again a generalized Reed-Solomon code), this gives a very natural way to think about constructing and analyzing codes.  Indeed, it is through this characterization that we prove (1) and (2).  Moreover, this characterization is useful to construct schemes for arbitrary RS codes.  
We give two further examples of applications to RS codes, beyond (1) and (2).  In the first, we give a non-trivial construction for a family of RS codes where the length $n$ of the code is small compared to $F$; allowing for larger field sizes can add flexibility in practice.  In our second example, we consider a specific RS code, used in the Facebook Analytics Hadoop cluster.  This particular code has been analyzed before~\cite{xoringelephants, scalarmds}, and using our characterization we are able to find a repair scheme (with the help of a computer) that out-performs the best known repair scheme for this code.

We remark that in subsequent work,~\cite{YB16} has used our characterization to find RS codes in this parameter regime (where $n$ is small compared to $F$) which are optimal regenerating codes.
\end{enumerate}
To the best of our knowledge, ours is the first work to systematically study the repair bandwidth of RS codes for general $k$.  In~\cite{scalarmds}, the authors give a framework for studying the repair bandwidth of linear MDS codes over finite fields, and as an example they analyze a few specific small RS codes.  Surprisingly---and this surprise was the inspiration for our work here---they show that for these small codes, one can do better than the naive repair scheme.  In this work we give a more general result, for all $k$, using different techniques.  We will survey the related work in the next section.

\paragraph{Organization.}
In Section~\ref{sec:prelim}, we set up notation and survey related work.  In Section~\ref{sec:results}, we state our results in more detail, and give an outline of the proofs, which are contained in Sections~\ref{sec:wlog},\ref{sec:upper}, and \ref{sec:lower}.  We conclude in Section~\ref{sec:conclusion} with some open questions.

\section{Set-up and Related Work}\label{sec:prelim}
In this section we set up our definitions for Reed-Solomon and regenerating codes, and survey related work.  We note that,
in the regenerating codes literature, it is common to use the Greek letters $\alpha,\beta,\gamma$ for parameters of the code.  We prefer to reserve Greek letters for elements of the field $F$ and use Roman letters (like $t,b$) the parameters.
For convenience, Table~\ref{fig:translation} in Appendix~\ref{app:notation} offers a translation between our notation and common notation in the regenerating codes literature.

\subsection{Codes and Reed-Solomon Codes}

A \em code \em $\cC$ over a field $F$ of length $n$ is a subset $\cC \subset F^n$.  We will view an element of $F^n$ as an $F$-valued function over a domain $A$ of size $n$.  Thus, a code is a collection of functions $\mathcal{F}$ from $A = \inset{\alpha_1,\ldots,\alpha_n}$ into $F$: the code $\cC \subseteq F^n$ determined by $\mathcal{F}$ and $A$ is 
\[ \cC = \inset{ (f(\alpha_1),f(\alpha_2), \ldots, f(\alpha_n) ) \suchthat f \in \mathcal{F} }. \]
In this work, we will often abuse notation and write $f \in \cC$, to mean that the evaluation vector $(f(\alpha_1),\ldots,f(\alpha_n))$ is in $\cC$.
The number of evaluation points $n$ is called the \em block length \em of the code.  
In this work, we study \em linear \em codes, i.e., those where $\mathcal{F}$ forms an $F$-vector space, and so $\cC$ forms a subspace of $F^n$.  For a linear code, the \em dimension \em $k$ is the dimension of this subspace, and the \em rate \em $r$ is defined as the ratio $k/n$.  
We refer to a function $f \in \mathcal{F}$ as a \em message, \em and a corresponding vector $(f(\alpha_1),\ldots,f(\alpha_n)) \in \mathcal{C}$ as a \em codeword. \em  For $c \in F^n$, we refer to the components $c_i \in F$ of $c$ as \em symbols. \em 

A Reed-Solomon code is the linear code formed when $\mathcal{F}$ is a set of low-degree polynomials, and $A \subset F$ is some set of evaluation points.
\begin{definition}\label{def:rs}
The Reed-Solomon code $RS(A,k) \subset F^n$ of dimension $k$ over a finite field $F$ with evaluation points $A = \inset{\alpha_1,\alpha_2,\ldots,\alpha_n} \subset F$ is the set
\[ RS(A,k) = \inset{ (f(\alpha_1), f(\alpha_2), \ldots, f(\alpha_n)) \suchthat f:F \to F \text{ is a polynomial of degree at most $k-1$ } }. \]
\end{definition}
Reed-Solomon codes are \em Maximum Distance Separable \em (MDS) codes, which means that any $k$ symbols (that is, evaluations of a polynomial $f$) can be used to recover the entire codeword (that is, $f$ itself).  
\begin{definition}
A linear code $\cC$, given by $\mathcal{F},A$, is Maximum Distance Separable (MDS) if the minimum distance of the code is the maximum possible, that is, if
\[ \min_{f \neq g \in \mathcal{F} } \inabs{ \inset{ \alpha \in A \suchthat f(\alpha) \neq g(\alpha) }}  = n - k + 1. \]
\end{definition}
In particular, in an MDS code, any $k$ symbols $f(\alpha)$ are enough to determine $f$ and hence the entire codeword.  Conversely, $k$ symbols are necessary to determine $f$: given only $k-1$ symbols, a remaining symbol $f(\alpha^*)$ could be any element of $F$.

Our characterization of linear repair schemes will go throught the \em dual code, \em $\cC^\perp$:
\begin{definition}
Suppose $\cC$ is a linear code given by $\cF$, $A$.  
The dual code $\cC^\perp$ of $\cC$ is
\[ \cC^\perp = \inset{ g:F \to F \suchthat \sum_{\alpha \in A} f(\alpha)g(\alpha) = 0 }. \]
\end{definition}
The dual of an MDS code is again MDS.

\begin{remark}[Non-standard notation for MDS codes]
While it is common to view Reed-Solomon codes as sets of functions $f:F^n \to F$, it is not standard to think of general MDS codes this way.  
In particular, we index the positions $1,\ldots,n$ of the codeword by evaluation points $\alpha_1,\ldots,\alpha_n$, even though for general MDS codes there may not be a natural choice of such evaluation points.
Our main theorems are general and apply to all MDS codes, but the primary motivation of this work is Reed-Solomon Codes.  Further, in the context of Reed-Solomon Codes, the choice of evaluation points is the crux of our constructions.   For these reasons, we stick with the Reed-Solomon-inspired notation throughout the paper. 
\end{remark}

\subsection{Exact Repair Problem and Regenerating Codes}

Recall the exact repair problem from the introduction: a file, consisting of $k$ blocks, is encoded into $n$ nodes.  The goal is to recover the contents of an erased node by downloading some information from the remaining nodes.
In the language of MDS codes as above, the file is a 
function $f \in \mathcal{F}$, which can be represented as $k$ symbols from a finite field $F$.
Each of the $n$ nodes is associated with an evaluation point
$\alpha \in A$, and it stores $f(\alpha)$.  For an arbitrary $\alpha^* \in A$ (corresponding to an erased node), the goal is to recover $f(\alpha^*)$ given some
information from $f(\alpha)$ for $\alpha \in A \setminus \inset{ \alpha^* }$.
Crucially, we may choose to download only \em part \em of each symbol
$f(\alpha) \in F$, meaning that a node may return fewer than $\log_2(|F|)$ bits of information when queried.   
More precisely, each node may return some number of \em sub-symbols. \em  A sub-symbol is an element of some ``base" set $B$ which is smaller than $F$---for example, $B = \{0,1\}$. 
%
While in principle a node's response can be an arbitrary function of its contents, in
this work we focus on \em linear \em repair schemes.  That is, we assume that
$B \leq F$ is a subfield, we view $F$ as a vector space over $B$, and we assume
that each node $\alpha$ may return any $B$-linear function of its contents $f(\alpha)$.
The $B$-linear
transformations from $F$ to  $B$ are precisely the \em trace functionals \em
$L_\gamma: F \to B$ given by
$ L_\gamma(\beta) = \tr_{F/B}( \gamma \beta ).$
Here, $\tr_{F/B}$ is the \em field trace \em of $F$ over $B$:
\begin{definition}
Let $F = GF(q^t)$ be a finite field extension of $B = GF(q)$ of degree $t$.  The \em field trace \em is defined as 
\[ \tr_{F/B}(\beta) = \beta + \beta^q + \beta^{q^2} + \cdots + \beta^{q^{t-1}}. \]
\end{definition}
Thus, in a linear repair scheme, the node corresponding to $\alpha$ 
returns zero or more elements of $B$ of the form $L_\gamma(f(\alpha))$. 
A linear exact repair scheme can then be described by the field elements $\gamma$ that are used in each trace functional, along with a (linear) repair algorithm.  We give a precise definition below.
\begin{definition}[Linear exact repair scheme]\label{def:repair}
Let $\cC$ be a linear code over $F$ of length $n$ and dimension $k$, given by a collection of functions $\mathcal{F}$ and a set of evaluation points $A$.  A \em linear exact repair scheme \em for $\cC$ over a subfield $B \leq F$ consists of the following.
\begin{itemize}
	\item For each $\alpha^* \in A$, and for each $\alpha \in A \setminus \inset{ \alpha^* }$, a set of queries $Q_\alpha(\alpha^*) \subset F$.
	\item For each $\alpha^* \in A$, a linear reconstruction algorithm that computes 
\[ f(\alpha^*) = \sum_i \lambda_i \nu_i \]
for coefficients $\lambda_i \in B$ and a basis $\nu_1,\ldots,\nu_t$ for $F$ over $B$, so that the coefficients $\lambda_i$ are $B$-linear combinations of the queries 
\[{ \bigcup_{\alpha \in A\setminus \alpha^*} \inset{ \tr_{F/B}( \gamma f(\alpha) ) \suchthat \gamma \in Q_\alpha(\alpha^*) } }.\]
\end{itemize}
The \em repair bandwidth \em $b$ of the exact repair scheme is the total number of sub-symbols in $B$ returned by each node $\alpha$:
\[ b = \max_{\alpha^* \in A} \sum_{\alpha \in A \setminus \inset{\alpha^*}} |Q_\alpha(\alpha^*)|. \]
The \em repair locality \em of the exact repair scheme is the number of $\alpha$ which are required to respond:
\[ d = \max_{\alpha^* \in A} \sum_{\alpha \in A \setminus \inset{ \alpha^* }} \ind{ Q_{\alpha}(\alpha^*) \neq \emptyset }. \]
We will define
\[ t = \log_{|B|}(|F|) \]
to be the dimension of $F$ as a vector space over $B$.  Thus, we can view each symbol from $F$ as a vector of $t$ sub-symbols from $B$.
\end{definition}

We illustrate the setup for Definition~\ref{def:repair} in Figure~\ref{fig:setup}.

\vspace{.22cm}

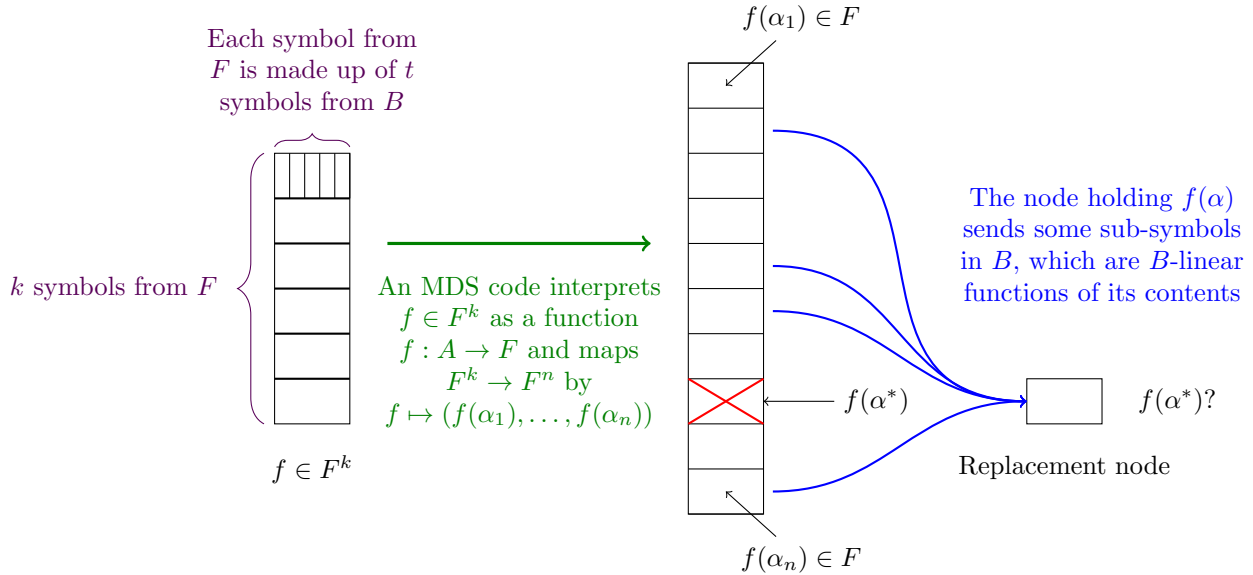
\begin{figure}
\begin{tikzpicture}[yscale=.6]
\draw (0,0) rectangle (1, 6);
\foreach \i in {1,...,5}
{
	\draw[thick] (0, \i) -- (1,\i);
}
\foreach \i in {1,...,5}
{
	\draw[thin] (\i/5, 5) -- (\i/5, 6);
}
\draw [decorate,decoration={brace,amplitude=10pt},xshift=-4pt,yshift=0pt,violet!70!black]
(0,0) -- (0,6) node [black,midway,xshift=-2cm, violet!70!black] {$k$ symbols from $F$};
\draw [decorate,decoration={brace,amplitude=5pt},xshift=0pt,yshift=4pt, violet!70!black]
(0,6) -- (1,6) node [black,midway,yshift=1cm, violet!70!black] {\cmp{3}{Each symbol from $F$ is made up of $t$ symbols from $B$}};
\node(a) at (.5,-1) {$f \in F^k$};
\draw[very thick,->,green!50!black] (1.5,4) to node[midway,below,yshift=-.3cm]{\cmp{3.8}{An MDS code interprets $f\in F^k$ as a function $f:A \to F$ and maps $F^k \to F^n$ by 
\\ {$f \mapsto (f(\alpha_1),\ldots, f(\alpha_n))$}}} (5,4);
\begin{scope}[xshift=5.5cm]
\draw (0,-2) rectangle (1,8);
\foreach \i in {-2,...,8}
{
	\draw (0,\i) -- (1, \i);
}
\node(c) at (1.5, 9) {$f(\alpha_1) \in F$};
\draw[->] (c) to (.5, 7.5);
\node(d) at (1.5,-3) {$f(\alpha_n) \in F$};
\draw[->] (d) to (.5, -1.5);
\node(e) at (2.5, .5) {$f(\alpha^*)$};
\draw[->] (e) to (1, .5);
\draw[thick,red] (0,0) -- (1,1);
\draw[thick,red] (0,1) -- (1,0);
\node(f1) at (1,3.5) {};
\node(f2) at (1,2.5) {};
\node(f3) at (1,6.5) {};
\node(f4) at (1,-1.5) {};
\end{scope}
\begin{scope}[xshift=10cm]
\draw (0,0) rectangle (1,1);
\node at (2,.5) {$f(\alpha^*)$?};
\node[anchor=north] at (.5,-.5) {\cmp{3}{Replacement node}};
\foreach \a in {1,2,3,4}{
	\draw[thick,blue,->] (f\a) to[out=0,in=180] (0,.5);
}	
\node[blue] at (1,4) {\cmp{4}{The node holding $f(\alpha)$ sends some sub-symbols in $B$, which are $B$-linear functions of its contents}};
\end{scope}
\end{tikzpicture}
\caption{Setup for a linear exact repair scheme.  $B \leq F$ is a subfield of $F$, and $F$ is a $t$-dimensional vector space over $B$.
For a Reed-Solomon code, the encoding maps a file $f \in (\beta_0,\beta_1,\ldots,\beta_{k-1}) \in F^k$ to evaluations of the polynomial $f(\bX) := \sum_i \beta_i \bX^i$.  The repair bandwidth is the number of symbols from $B$ that the replacement node needs to download in order to reconstruct $f(\alpha^*)$.}
\label{fig:setup}
\end{figure}

\paragraph{What we care about.}  There are several parameters of interest in the use of MDS regenerating codes for storage.  
The three that we focus on in this work are the rate $k/n$, the repair bandwidth $b$, and the size of the base field $B$. We would like the rate to be as large as possible, ideally approaching $1$; this means that we minimize storage overhead.  
We would also like to minimize the number of \em bits \em $b\log_2(|B|)$ downloaded by the replacement node; this means we would like to minimize the communication from the remaining nodes to the replacement node in the repair process.  

Before we discuss related work and the use of Reed-Solomon codes for the exact repair problem, we make a few remarks about our definitions and goals, and their relationship to the regenerating codes literature.
\begin{remark}[Measuring bandwidth in bits]\label{rem:bits}
We focus on the quantity $b\log_2(|B|)$ rather than on just $b$ for the following reason.  When $B = F$, it is trivial to obtain $b = k$, which is clearly optimal.  However, this is also clearly not a good solution, as it is the same as the traditional RS approach from the introduction.  Focusing on $b\log_2(|B|)$ means that we always measure bits, rather than symbols of some possibly-large subfield $B$.
\end{remark} 

\begin{remark}[Direction of communication]
The definition of repair bandwidth above only counts communication from the remaining nodes to the replacement node.  The astute reader will have noticed that the replacement node must also communicate to the remaining nodes!  Indeed, the remaining nodes must know the identity $\alpha^*$ of the erased node (or at least know what function of their contents they are supposed to return).  While this replacement-to-remaining-nodes direction of communication is important, in practice the cost of this communication is negligable compared to the remaining-nodes-to-replacement direction that is captured in the definition of repair bandwidth.  We elaborate more on this point in Appendix~\ref{app:justification}.  For now, we just point out the regenerating codes literature focuses on this one-way definition of repair bandwidth, and we also adopt this definition in our work. 
\end{remark}

\begin{remark}[MSR codes]\label{rem:msr}
In the regenerating codes literature, the size of the blocks in the file needn't be the same as the storage capacity of the nodes, and there is a beautiful theory investigating the trade-offs this involves.  Because we are interested in Reed-Solomon codes, which have the same message alphabet and codeword alphabet, these sizes are the same.  In the regenerating codes terminology, this means we are working with \em minimum storage regenerating \em (MSR) codes.\footnote{These codes are referred to as \em minimum storage \em because the storage in each of the $n$ nodes is as small as possible, given that any $k$ nodes should be able to reconstruct the message.}  In this work we restrict our discussion to this setting.
\end{remark}

\begin{remark}[A complicated landscape]
There are many figures of merit and variations on the exact repair problem.  For example, in addition to rate and bandwidth, we may care about locality; we may care about multiple erasures; we may not need to reconstruct the erased node exactly, but simply want to maintain the MDS property (this is called \em functional repair\em); we may want to leave the MSR parameter regime; and so on.  There is a growing body of work addressing these and other trade-offs, and the lay of the land is still not fully understood.  The reader is referred to the survey~\cite{survey} and the very helpful \em Erasure Coding for Distributed Storage Wiki\em~\cite{wiki} for more details about these and other variants.
\end{remark}

\subsection{Repair Bandwidth of Reed-Solomon Codes}

Reed-Solomon codes are commonly used for storage, but as mentioned above the traditional strategy (which has $b\log_2(|B|) = k\log_2(|F|)$) is not a good idea for the exact repair problem.  
However, the traditional strategy is not the best one can do!
In~\cite{scalarmds}, Shanmugan et al. develop a general framework for
studying the repair bandwidth of \em scalar MDS codes\em---that is, codes whose
symbols naturally come from some field $F$ rather than being constructed
specifically as vectors over $B$.  As one of their examples, they show that for a few specific Reed-Solomon codes, one can do better than the naive scheme.

More precisely, \cite{scalarmds} adapts techniques from \em interference alignment \em (which have been previously used to construct good regenerating codes) to the scalar MDS setting.  For general MDS codes with $k = n-2$, they give a polynomial-time algorithm which will find the optimal linear systematic repair scheme returning a single symbol from the subfield $B$.  They apply this to find optimal linear exact systematic repair schemes for a $(5,3)$ and $(6,4)$-Reed-Solomon codes\footnote{An $(n,k)$-RS code has block length $n$ and dimension $k$.}, and they find non-trivial systematic repair schemes for the $(14,10)$ Reed-Solomon code used in a module for the Apache Hadoop Distributed File System which is currently deployed by Facebook.

There have been works which use RS codes as a building block for codes for distributed storage and related problems~\cite{xoringelephants,HPZV13,HZM12,RSKR09,TPD13,TB14,TB14b,HLKB15}.  These works modify RS codes by, for example, adding parity checks, taking subcodes, folding, concatenating with other codes, and so on, but to the best of our knowledge, only the work of~\cite{scalarmds} described above addresses the repair bandwidth of Reed-Solomon codes themselves.
Before we describe the rest of the literature surrounding exact recovery, we note two differences between our approach and that of~\cite{scalarmds}.
\begin{itemize}
\item First, in~\cite{scalarmds}, the proof applies only for $k = n-2$, while our approach works for all $(n,k)$.  On the other hand, their approach works for any MDS code, while ours is tailored for Reed-Solomon codes.
\item Second, \cite{scalarmds} considers exact repair of \em systematic nodes \em only.  That is, $k$ of the $n$ storage nodes hold the original message, and the rest are viewed as parity checks; only these $k$ special nodes are required to be repairable.  In contrast, our approach guarantees recovery of all $n$ nodes.
\end{itemize}
As pointed out in~\cite{scalarmds}, understanding the repair bandwidth of Reed-Solomon codes is an important problem, even if RS codes are not the best codes available.  Indeed, these codes are implemented in practice in distributed storage systems (the example from~\cite{scalarmds} is the HDFS-RAID module, which we will return to in Section~\ref{sec:hdfs}), and it may be easier to implement improved algorithms on existing systems rather than replace the system.
Pinning down the repair bandwidth of Reed-Solomon codes was asked as an open question by Dimakis et al. in~\cite{survey}.

\subsection{Existing Results for the Exact Repair Problem for General MDS Codes}
In order to set expectations, we briefly survey the upper and lower bounds available for exact repair using MDS codes (not necessarily RS codes). 
There are two main parameter regimes, depending on the parameter $t$.
This parameter (which is the number of sub-symbols per symbol, or the degree of $F$ over $B$) controls the level of \em subpacketization \em in the regenerating code.  
The first parameter regime, more commonly studied for regenerating codes, is when $t$ is (very) large compared to $n-k$.  
The second parameter regime, more natural for RS codes, is when $t$ is small compared to $n-k$.  Both settings have their advantages.  When $t$ is large, each symbol can be sub-divided further (we have more subpacketization), and as we will see this allows for better bandwidth guarantees.  On the other hand, when $t$ is small, the field extension $F$ over $B$ is smaller, and this is easier to work with in practice.
%

In this work we consider both parameter regimes.  Our main focus is constant-rate RS codes with $A = F$, and so $t = \log_2(n)$ is small compared to $n-k$. 
However, our framework also works for RS codes with $A \subset F$ and with $n-k$ very small, and we give examples of constructions when $t$ is large compared to $n-k$ as well.

\paragraph{Regime 1: large $t$.}
When $t$ is sufficiently large, it is known that the ``correct" answer for the repair bandwidth is 
\begin{equation}\label{eq:bigt}
 b = \frac{ td }{d + 1 - k}, \tag{$\star$}
\end{equation}
The lower bound on $b$ is a fundamental result of~\cite{DGWW10,WDR07}, and actually holds for \em functional \em repair as well as exact repair.\footnote{In the functional repair problem, the replacement node needn't be a copy of the lost node; rather it just must maintain the MDS property.  For some applications this is enough.  However, for several reasons it is also useful to study the exact repair problem~\cite{SR10}.  Further, for us, only the exact repair problem makes sense given that we want to study a fixed code $RS(A,k)$.}
For the upper bound, it is shown in~\cite{SR10b, CJM13} that as $t \to \infty$ (much faster than $n$), the exact repair bandwidth can approach \eqref{eq:bigt}.
However, for this result, $t$ must scale exponentially in $n$.  It is conjectured that this exponential scaling is necessary~\cite{twb14}, but the best that is known is that $t$ must be at least $k^2$ in order for~\eqref{eq:bigt} to hold; for very high-rate codes, with $k = n-O(1)$, we do know that $t \geq \exp(\sqrt{n})$ is required~\cite{gtc}.  
There are also several schemes acheiving \eqref{eq:bigt} exactly for particular parameter settings and/or systematic repair only, and for large $t$~\cite{CHJL11, SR10, SRKR09, PDC13, CHL11, TWB13}.

\paragraph{Regime 2: small $t$.}
When $t$ is small compared to $n-k$, on the other hand, 
it is clear that \eqref{eq:bigt} cannot be met.  Indeed, 
\begin{equation}\label{eq:trivial}
b \geq k + t - 1 \tag{$\star\star$}
\end{equation}
 is a trivial lower bound on the exact repair bandwidth for any MDS code.\footnote{To see this, imagine that we download only a single sub-symbol from each of $k-1$ nodes and are given the remaining symbols for free.  Because the code is MDS, the final symbol could be anything; thus we need to read at at least one more symbol's worth of information---or at least $t$ more sub-symbols---to determine it.}  Thus,
if $t < n-k$,
 we have $b \gg \frac{td}{d + 1 - k}$.
In this regime, we must have the ratio $b/t$ tend to infinity.  However, it is still the case that a bound of $b = t + k - 1$ is much better than the naive bound of $b = tk$.

We are only aware of two works addressing the exact repair problem when $t$ is small compared to $n-k$.  The first is~\cite{WD09}, who give a scheme with bandwidth $(k-1)t + 1$.  Since the naive scheme has bandwidth $kt$, this is a slight improvement.  The second work is~\cite{SRKR09}.  There the authors give optimal schemes, meeting \eqref{eq:trivial} for $t \geq k-2$.  However, for $t < k$ these results hold only for systematic nodes.  They also show that, when only one sub-symbol is downloaded from every node and the reconstruction algorithm is linear, \eqref{eq:trivial} cannot be met for $t \leq k - 3$.

We give a more detailed summary of known results for the exact repair problem for MDS codes, and compare them to our results for Reed-Solomon codes, in Table~\ref{fig:exactmsr} in Appendix~\ref{app:litreview}.   We outline our results in more detail in the next section.

\section{Results Overview}\label{sec:results}
Our main result is pinning down the best exact repair bandwidth of Reed-Solomon codes with $A = F$ that can be acheived with linear schemes (Definition~\ref{def:repair}).   This will follow for a characterization of linear exact repair schemes for RS codes, which is given in Section~\ref{sec:wlog}.

To formulate our characterization, we first show in Theorem~\ref{thm:wlog} that any linear exact repair scheme proceeds roughly as follows.
First, we write the erased data $f(\alpha^*)$ as a linear combination of the available data $f(\alpha)$:
\[ \zeta f(\alpha^*) = \sum_{\alpha \in A \setminus \alpha^*} \mu_{\alpha,\zeta}(\alpha^*) f(\alpha), \]
and we may do this for several different $\zeta \in F$.
Next, we take the trace of both sides: 
\[ \tr_{F/B}( \zeta f(\alpha^*) ) = \sum_{\alpha \in A \setminus \alpha^*} \tr_{F/B}(\mu_{\alpha,\zeta}(\alpha^*) f(\alpha) ). \]
If, for each $\alpha$, the node corresponding to $\alpha$ delivers $\tr_{F/B}(\mu_{\alpha,\zeta}(\alpha^*) f(\alpha) )$, we can recover $\tr_{F/B}(\zeta f(\alpha^*))$.
If we do this for enough different $\zeta$'s we can recover $f(\alpha^*)$.  Thus, our goal is to find $\mu$'s and $\zeta$'s so that there are many collisions between the multipliers $\mu_{\alpha,\zeta}(\alpha^*)$ that a given node $\alpha$ is responsible for. 
As stated, this appears to be a daunting task.
However, we show in Theorem~\ref{thm:polynomials} that this task is equivalent to the problem of finding some nice polynomials over $F$, and this will give us our characterization.

In Section~\ref{sec:upper}, we use this characterization with trace polynomials to obtain an exact repair scheme for high-rate Reed-Solomon codes which use the whole field as evaluation points.
More precisely, we prove the following theorem.
\begin{theorem}\label{thm:fulllen}
Let $B \leq F$ be any subfield of $F$, and let $k = (1 - 1/|B|)|F|$.  Then the Reed-Solomon code $RS(F, k)$ of rate $1 - 1/|B|$ which uses the entire field $F$ as evaluation points admits a linear exact repair scheme over $B$ with repair bandwidth $n-1$.
\end{theorem} 
As per Remark~\ref{rem:bits}, it is 
instructive also to write this in terms of bits.  Returning a symbol of $B$ is equivalent to returning $\log_2(B)$ bits, and we have the following corollary.
\begin{cor}\label{cor:binary}
Suppose that $F$ has characteristic $2$.
Let $B \leq F$ be a subfield and let $\eps = |B|^{-1}$.  
Then there is a linear exact repair scheme for $RS(F, (1-\eps)|F|)$ over $GF(2)$ with repair bandwidth $(n-1)\log_2(1/\eps)$. 
\end{cor}
This scheme is nearly optimal for MDS codes with any repair scheme; the lower bound is $k + t - 1$ sub-symbols, and Theorem~\ref{thm:fulllen} uses $k/(1-\eps) - 1$ subsymbols.  However, our second contribution is to prove an even stronger lower bound for linear repair schemes.
More precisely, in Section~\ref{sec:lower}, we show that Corollary~\ref{cor:binary} is optimal, even up to the leading constants, 
for linear schemes.
\begin{theorem}	\label{thm:optimal}
Let $\mathcal{C}$ be an MDS code of dimension $k$ with evaluation points $A$ over a field $F$.  Let $B \leq F$ be a subfield.  Any linear repair scheme for $\cC$ over $B$ must have bandwidth (measured in subsymbols of $B$) at least
\[ b \geq (n-1) \log_{|B|} \inparen{ \frac{n-1}{n-k} }. \]
In particular, the bandwidth (measured in bits) for any linear repair scheme for an MDS code with rate $ 1- \eps$ over any base field $B$ is at least
\[ b\log_2(|B|) \geq (n-1) \log_2\inparen{ \frac{1}{\eps} \inparen{1 - \nicefrac{1}{n}} }. \]
\end{theorem}

We also give a few other examples of how to use our characterization for RS codes with $A \neq F$.  In Section~\ref{sec:bigt}, we give an example of an RS code with non-trivial bandwidth when $|F|$ might be arbitrarily larger than $n$.  In Section~\ref{sec:hdfs}, we use our characterization, along with a computer search, to find a scheme which improves the result of~\cite{scalarmds} for the code used in the HDFS-RAID module~\cite{HDFS}; this module is currently deployed in the Facebook Hadoop Analytics cluster~\cite{xoringelephants}.

\subsection{Subsequent work}
After a preliminary version of this paper was released, it was shown in~\cite{YB16} that in fact our framework can be used to obtain good repair schemes for RS codes in the ``large $t$" regime, meeting the cut-set bound~\ref{eq:bigt}.  
More precisely, for any field $B$, and for any $n,k$, they show how to choose $n$ evaluation points $A \subset \F_{|B|^t}$, for $t = (n-k)^n$, so that $RS(A,k)$ has repair bandwidth $b = t\inparen{\frac{n-1}{n-k}}$.
Thus, the final take-away should be that, in all parameter regimes, Reed-Solomon codes are competetive (in terms of bandwidth) with regenerating codes in both parameter regimes!

\section{Characterization of linear repair schemes for MDS codes}\label{sec:wlog}
In this section, 
we give a characterization of linear exact repair schemes for MDS codes.  The following theorem says that a linear exact repair scheme for a $k$-dimensional MDS code is equivalent to being able to find, for each $\alpha^* \in A$, a set $\mathcal{P}(\alpha^*)$ of dual codewords $p \in \cC^\perp$ so that $\inset{ p(\alpha) : p \in \mathcal{P}(\alpha^*)}$ spans a low-dimensional subspace for $\alpha \neq \alpha^*$, and spans a high-dimensional subspace for $\alpha = \alpha^*$.
\begin{theorem}\label{thm:polynomials}
Let $B \leq F$ be a subfield so that the degree of $F$ over $B$ is $t$, and let $A \subset F$ be any set of evaluation points.
Let $\cC \subset F$ be an MDS code of dimension $k$, with evaluation points $A$.
The following are equivalent.
\begin{itemize}
\item[(1)] 
There is a linear repair scheme for $\cC$ over $B$ with bandwidth $b$.
\item[(2)]
For each $\alpha^* \in A$, there is a set $\mathcal{P}(\alpha^*) \subset \cC^\perp$ of size $t$, so that 
\[ \dim_B\inparen{\inset{ p(\alpha^*) \suchthat p \in \mathcal{P}(\alpha^*) }} = t, \]
and the sets $\inset{ p(\alpha) \suchthat p \in \mathcal{P}(\alpha^*)}$ for $\alpha \neq \alpha^*$ satisfy
\[ b \geq \max_{\alpha^* \in A} \sum_{\alpha \in A \setminus \inset{\alpha^*}} \dim_B\inparen{ \inset{ p(\alpha) \suchthat p \in \mathcal{P}(\alpha^*) } } . \]
\end{itemize}
\end{theorem}

To prove Theorem~\ref{thm:polynomials}, we begin by showing that any linear repair scheme for MDS codes may as well have a particularly nice form.  More precisely, we will show that it may as well have the form of Algorithm~\ref{algo:wlog}.
\begin{algorithm}
\caption{Framework of generic linear repair scheme for an MDS code $\cC \subset F^n$ over a subfield $B$, so that the degree of $F$ over $B$ is $t$.}
\label{algo:wlog}
\KwIn{A set $A$, a failed node $\alpha^* \in A$, and access to linear queries of the form $\tr_{F/B}(\gamma \cdot f(\alpha))$ for $\alpha \in A \setminus \inset{\alpha^*}$, for some $f \in \cC$.}
\KwOut{The value $f(\alpha^*)$}
 Choose a set $Z \subset F$ of size $t$, which has full rank over $B$.

Choose coefficients $\mu_{\zeta,\alpha}(\alpha^*)$ for $\alpha \in A\setminus \inset{ \alpha^*}$ and $\zeta \in Z$ so that
\begin{equation}\label{eq:thing1}
 \zeta f(\alpha^*) = \sum_{\alpha \in A \setminus \inset{ \alpha^* } } \mu_{\zeta, \alpha}(\alpha^*) f(\alpha). 
\end{equation}

\For {$\zeta \in Z$}{
	Let $\tilde{Q}_\alpha(\alpha^*) \subset F$ be any spanning set for $\inset{ \mu_{\zeta,\alpha}(\alpha^*) \suchthat \zeta \in Z}$ over $B$, and query
	\[ \tr_{F/B}( \gamma \cdot f(\alpha) ) \qquad \text{ for } \gamma \in \tilde{Q}_\alpha(\alpha^*). \]

	Using the $B$-linearity of $\tr_{F/B}$, compute $\tr_{F/B}( \mu_{\zeta,\alpha}(\alpha^*) f(\alpha) )$ for each $\alpha \in A \setminus{\alpha^*}$.
	
	Construct $\tr_{F/B}(\zeta \cdot f(\alpha^*))$ from the identity
\[ \tr_{F/B}(\zeta \cdot  f(\alpha^*)) = \sum_{\alpha \in A \setminus \inset{ \alpha^*} } \tr_{F/B}( \mu_{\zeta,\alpha}(\alpha^*) f(\alpha), \]
	which follows from taking the trace of both sides of \eqref{eq:thing1}.
}

Compute $f(\alpha^*)$ from the data $\inset{\tr_{F/B}(\zeta \cdot f(\alpha^*)) \suchthat \zeta \in Z }$.  More precisely, since $Z = \inset{ \zeta_1,\ldots,\zeta_t}$ are a basis for $F$ over $B$, let $V = \inset{ \nu_1,\ldots,\nu_t}$ be the dual basis.  Then
\[f(\alpha^*) = \sum_{i=1}^t \tr_{F/B}( \zeta_i f(\alpha^*) ) \nu_i.\]

\end{algorithm}
By inspection, it is clear that Algorithm~\ref{algo:wlog} is indeed a linear repair scheme for $\cC$, for any choice of basis $Z$ for $F/B$ and for any coefficients $\mu_{\zeta,\alpha}(\alpha^*)$ so that \eqref{eq:thing1} holds.  We record this fact in the following proposition.
\begin{proposition}\label{prop:itworks}
Algorithm~\ref{algo:wlog} is a linear repair scheme for $\cC$ over $B$ with bandwidth
\[ b = \max_{\alpha^* \in A} \sum_{\alpha \in A \setminus \inset{\alpha^*}} \dim_B(\inset{ \mu_{\zeta,\alpha}(\alpha^*) \suchthat \zeta \in Z }). \]
\end{proposition}
Moreoever, any linear repair scheme can be written in the form of Algorithm~\ref{algo:wlog}.
Proposition~\ref{thm:wlog} and the ensuing Corollary~\ref{cor:wlog} make this precise.  
\begin{proposition}\label{thm:wlog}
Suppose there is a linear repair scheme for an MDS code $\cC \subset F$ over $B \leq F$, given by query sets $Q_\alpha(\alpha^*)$ and a linear repair algorithm, as in Definition~\ref{def:repair}.
Then there is a basis $Z$ for $F$ over $B$ so that the following holds.
For each $\alpha^* \in A$ and $\alpha \in A \setminus{\alpha^*}$, and $\zeta \in Z$, there are coefficients $\mu_{\alpha,\zeta}(\alpha^*)$ so that
\[ \zeta f(\alpha^*) = \sum_{\alpha \in A \setminus\inset{\alpha^*}} \mu_{\alpha,\zeta}(\alpha^*) f(\alpha) \]
for all $f \in \cC$,  and so that for all $\alpha \neq \alpha^* \in A$,
\[ \inset{ \mu_{\alpha, \zeta}(\alpha^*) } \subseteq \spn_B( Q_\alpha(\alpha^*) ). \] 
\end{proposition}
\begin{proof}
By assumption, the linear repair algorithm computes coefficients $\lambda_i \in B$ so that
\[ f(\alpha^*) = \sum_{i=1}^t \lambda_i \nu_i \]
for some basis $\nu_i$ of $F$ over $B$. 
Since the $\lambda_i$ are $B$-linear functions of the queries, they are of the form
\[\lambda_i = \sum_{\alpha \neq \alpha^*} \sum_{\gamma \in Q_\alpha(\alpha^*)} \beta_{\alpha,\gamma,i} \cdot \tr_{F/B}( \gamma \cdot f(\alpha) ) \]
for some coefficients $\beta_{\alpha,\gamma,i} \in B$. 
Let $\zeta_1,\ldots,\zeta_t$ be the dual basis for $\nu_1,\ldots, \nu_t$, so that $\tr_{F/B}(\zeta_i\nu_\ell) = \ind{i = \ell}$.  Then 
for any $i \leq t$,
\begin{equation}\label{eq:wlog}
\tr_{F/B}(\zeta_i f(\alpha^*)) = 
\lambda_i = 
\sum_{\alpha \neq \alpha^*} \tr_{F/B} \inparen{ \sum_{\gamma \in Q_\alpha(\alpha^*)} \beta_{\alpha,\gamma,i} \gamma f(\alpha) }
=: \sum_{\alpha \neq \alpha^*} \tr_{F/B} \inparen{ \mu_{\alpha,\zeta_i}(\alpha^*) f(\alpha) },
\end{equation}
where \eqref{eq:wlog} is defining the coefficients $\mu_{\alpha,\zeta_i}(\alpha^*)$.
Equation \eqref{eq:wlog} holds for all $f \in \cC$; since $\cC$ is a linear code, it holds also for the function $\gamma \cdot f(\bX)$ for $\gamma \in F$.   This implies that for all $f \in \cC$, and for all $\gamma \in F$, we have
\[ \tr_{F/B}\inparen{ \gamma \cdot \zeta_i f(\alpha^*)} = \tr_{F/B} \inparen{ \gamma \cdot \sum_{ \alpha \in A \setminus \inset{\alpha^*}}{\mu_{\alpha, \zeta_i}(\alpha^*)}\cdot f(\alpha)},\]
which in turn implies 
that for all polynomials $f \in F[\bX]$ of degree less than $k$, 
\[ \zeta_i f(\alpha^*) = \sum_{ \alpha \in A \setminus \inset{\alpha^*}}{\mu_{\alpha, \zeta_i}(\alpha^*)}\cdot f(\alpha).\]
Thus, we have a linear equation of the form required for each $\zeta_i \in Z$.
  Finally, we observe that the coefficients $\mu_{\alpha,\zeta_i}(\alpha^*)$
 live where they are supposed to.  We have 
\[  \mu_{\alpha,\zeta_i}(\alpha^*) = \sum_{\gamma \in Q_{\alpha}(\alpha^*)} \beta_{\alpha,\gamma,i} \cdot \gamma, \]
and so
\[ \inset{ \mu_{\alpha,\zeta_\ell}(\alpha^*) \suchthat \ell = 1,\ldots,t } \subseteq \spn_B{ Q_{\alpha}(\alpha^*) }, \]
as desired.
\end{proof}
\begin{cor}\label{cor:wlog}
Let $B \subset F$ be a subfield and let $A \subset F$ be any set of evaluation points.  Let $k \leq |A|$ be any integer.  
Let $\cC \subset F$ be an MDS code with evaluation points $A$ and dimension $k$.
The following are equivalent.
\begin{itemize}
\item[(1)] There is a linear repair scheme for $\cC$ over $B \leq F$ with bandwidth $b$.
\item[(2)] There is a linear repair scheme for $\cC$ over $B$ of the form of Algorithm~\ref{algo:wlog}, with bandwidth at most $b$.
\end{itemize}
\end{cor}
\begin{proof}
The fact that (2) implies (1) follows from Proposition~\ref{prop:itworks}.  To show that (1) implies (2), 
suppose that there is a linear repair scheme for $RS(A,k)$ with query sets $Q_\alpha(\alpha^*)$.
 Choose the basis $Z$ and the coefficients $\mu_{\alpha,\zeta}(\alpha^*)$ in Algorithm~\ref{algo:wlog} as guaranteed by Proposition~\ref{thm:wlog}.  
By Proposition~\ref{prop:itworks}, the bandwidth of Algorithm~\ref{algo:wlog}, instantiated this way, is
\[ b = \max_{\alpha^* \in A} \sum_{\alpha \in A \setminus \inset{\alpha^*}} \dim_B(\inset{ \mu_{\alpha,\zeta}(\alpha^*) \suchthat \zeta \in Z}). \]
By Proposition~\ref{thm:wlog}, we have, for all $\alpha^*$,
\[ \inset{ \mu_{\alpha,\zeta}(\alpha^*) \suchthat \zeta \in Z } \subseteq \spn_B\inparen{ Q_\alpha(\alpha^*) }, \]
and so for all $\alpha^*$,
\[ \dim_B( \inset{ \mu_{\alpha, \zeta}(\alpha^*) \suchthat \zeta \in Z} ) \leq \inabs{ Q_\alpha(\alpha^*) }. \]
The corollary follows.
\end{proof}

Corollary~\ref{cor:wlog} says that coming up with an exact repair scheme for an MDS codes is equivalent to the problem of coming up with the basis $Z$ and the coefficients $\mu_{\zeta,\alpha}(\alpha^*)$.
It is not hard to see that finding such coefficients is equivalent to finding nice dual codewords.

\begin{obs}\label{prop:polynomials}
Fix a set $A \subset F$ with $|A| = n$, a subfield $B \leq F$ so that $F$ has degree $t$ over $B$, and an integer $k < n$.
Fix $\alpha^* \in A$ and numbers $d_\alpha \leq t$ for each $\alpha \in A \setminus \inset{\alpha^*}$.
Let $\cC \subset F$ be an MDS code with evaluation points in $A$.
The following are equivalent.
\begin{itemize}
\item[(1)]
There is a basis $Z$ for $F$ over $B$ and coefficients 
\[ \inset{\mu_{\alpha, \zeta}(\alpha^*) \suchthat \alpha \in A \setminus \inset{\alpha^*}, \zeta \in Z} \] 
so that for all $f \in \cC$, 
for all $\zeta \in Z$,
\begin{equation}\label{eq:grs}
 \zeta f(\alpha^*) = \sum_{\alpha \in A \setminus \inset{\alpha^*} } \mu_{\alpha, \zeta}(\alpha^*) f(\alpha),
\end{equation}
and for all $\alpha \in A \setminus \inset{ \alpha^*}$,
\[ \dim_B(\inset{ \mu_{\alpha, \zeta}(\alpha^*) \suchthat \zeta \in Z }) = d_\alpha. \]
\item[(2)] There is a set $\mathcal{P}(\alpha^*) \subset \cC^\perp$ of size $t$, so that
\[ \dim_B(\inset{ p(\alpha^*) \suchthat p \in \mathcal{P}(\alpha^*) } ) = t \]
and for all $\alpha \in A \setminus \inset{\alpha^*}$,
\[ \dim_B(\inset{ p(\alpha) \suchthat p \in \mathcal{P}(\alpha^*) }) = d_\alpha. \]
\end{itemize}
\end{obs}
\begin{proof}
This follows from the definition of duality. For the $(1)\Rightarrow (2)$ implication, given $Z$ and $\mu_{\alpha,\zeta}$, define $p_\zeta:F \to F$ by
$p_{\zeta}(\alpha^*) = -\zeta$ and $p_{\zeta}(\alpha) = \mu_{\alpha,\zeta}(\alpha^*)$, and let $\mathcal{P}(\alpha^*) = \inset{p_\zeta \suchthat \zeta \in Z}$.  For the other direction, given $\mathcal{P}(\alpha^*) = \inset{ p_1,\ldots,p_t} \subset \cC^\perp$, define $\zeta_i = p_i(\alpha^*)$ and let $Z = \inset{ \zeta_1,\ldots,\zeta_t}$.  Then let $\mu_{\alpha, \zeta_i}(\alpha^*) = p_i(\alpha)$. 
\end{proof}
Together, Observation~\ref{prop:polynomials} and Corollary~\ref{cor:wlog} prove Theorem~\ref{thm:polynomials}.

Finally, we apply the reasoning above to Reed-Solomon codes in particular, using the fact that the dual of a Reed-Solomon code is again a generalized Reed-Solomon code.
A \em generalized \em Reed-Solomon (GRS) code is the same as a Reed-Solomon code, except that there is an additional vector of multipliers that specify it:
\begin{definition}\label{def:grs}
The generalized Reed-Solomon (GRS) code of dimension $k$ with evaluation points $A = \inset{ \alpha_0,\ldots,\alpha_{n-1} } \subset F$ and multiplier vector $\lambda \in (F^*)^n$ is given by
\[ GRS(A,k,\lambda) = \inset{ ( \lambda_i f(\alpha_i ) )_{i=0}^{n-1} \suchthat f \in F[\bX], \deg(f) < k }. \]
\end{definition}
It is well-known (see, for example,~\cite[Thm.\,4 in Ch.10, $\S$8]{macwilliams-sloane}) that the dual of a Reed-Solomon code $RS(A,k)$ is a \em generalized \em Reed-Solomon (GRS) code $GRS(A,k,\lambda)$, where the multiplier $\lambda_i$ is given by
\[ \lambda_i = \prod_{j \neq i} (\alpha_i - \alpha_j)^{-1}. \]
Since these multipliers do not affect the dimension of the set $\inset{ p(\alpha)) \suchthat p \in \mathcal{p}(\alpha^*) }$ in the statement of Theorem~\ref{thm:polynomials}, we may leave them out, and we have the following corollary for Reed-Solomon codes.
\begin{cor}\label{cor:polynomials}
Let $B \leq F$ be a subfield so that the degree of $F$ over $B$ is $t$, and let $A \subset F$ be any set of evaluation points.
The following are equivalent.
\begin{itemize}
\item[(1)] 
There is a linear repair scheme for $RS(A,k)$ over $B$ with bandwidth $b$.
\item[(2)]
For each $\alpha^* \in A$, there is a set $\mathcal{P}(\alpha^*) \subset F[\bX]$ be a set of $t$ polynomials of degree less than $n-k$, so that 
\[ \dim_B\inparen{\inset{ p(\alpha^*) \suchthat p \in \mathcal{P}(\alpha^*) }} = t, \]
and the sets $\inset{ p(\alpha) \suchthat p \in \mathcal{P}(\alpha^*)}$ for $\alpha \neq \alpha^*$ satisfy
\[ b \geq \max_{\alpha^* \in A} \sum_{\alpha \in A \setminus \inset{\alpha^*}} \dim_B\inparen{ \inset{ p(\alpha) \suchthat p \in \mathcal{P}(\alpha^*) } } . \]
\end{itemize}
Moreover, suppose that (2) holds.  Then the linear repair scheme in (1) is given by Algorithm~\ref{algo:wlog}, with coefficients
\[ \mu_{\alpha,\zeta}(\alpha^*) = p(\alpha) \cdot \frac{ \prod_{\beta \in A \setminus \inset{\alpha^*}} (\alpha^* - \beta) }{\prod_{\beta \in A \setminus \inset{\alpha}} (\alpha - \beta) } \]
and the basis $Z$ given by
\[ Z = \inset{ p(\alpha^*) \suchthat p \in \mathcal{P}(\alpha^*)}. \]
\end{cor}

Thus, the task of finding repair schemes for RS codes boils down to choosing some particularly nice polynomials $\mathcal{P}$.  In the next section, we show several examples of how to pick these polynomials.  In particular, for $A=F$ we choose $\mathcal{P}$ to be trace polynomials, and obtain an optimal linear repair scheme for $RS(F,k)$ for any $k$.

\section{Constructions for RS codes}\label{sec:upper} 
By Corollary~\ref{cor:polynomials}, a linear repair scheme for Reed-Solomon codes can be specified by choosing evaluation points $A$, and, for each $\alpha^* \in A$, a set $\mathcal{P}(\alpha^*)$ of polynomials.  
In this section, we will make these choices in a few different parameter regimes.  First in Section~\ref{sec:fulllen} we will choose $A = F$, and prove Theorem~\ref{thm:fulllen}, giving an optimal linear repair scheme for general high-rate RS codes.  Next, in Section~\ref{sec:bigt}, we will give an example of a code where $n \ll |F|$ is much smaller.  
Finally in Section~\ref{sec:hdfs}, we will consider a concrete example, and give an improved repair scheme for the specific RS code used in the HDFS-RAID module, deployed by Facebook and studied in~\cite{xoringelephants, scalarmds}.

\subsection{When $A=F$: Optimal Linear Repair Schemes}\label{sec:fulllen}
In this section, we will choose $A = F$ to be the entire field and prove Theorem~\ref{thm:fulllen}.   For the reader's convenience, we restate the theorem below.
\begin{theorem*}[Theorem~\ref{thm:fulllen}]
Let $B \leq F$ be any subfield of $F$, and choose $k \leq n(1 - 1/|B|)$.  Then there is a linear exact repair scheme for $RS(F, k)$
with repair bandwidth $n-1$ over $B$.
\end{theorem*} 
\begin{proof}[Proof of Theorem~\ref{thm:fulllen}]
We will choose a set $\mathcal{P}(\alpha^*)$ of polynomials of degree less than $n/|B|$ for each $\alpha^* \in F$, so that the conditions of Corollary~\ref{cor:polynomials} hold. 
Fix any basis $Z \subset F$ for $F$ over $B$, and choose
\[ \mathcal{P}(\alpha^*) = \inset{ \frac{ \tr_{F/B} \inparen{ \zeta ( \bX - \alpha^* ) }}{ \bX - \alpha^* } \suchthat \zeta \in Z}. \]
Notice that these indeed have degree $|B|^{t-1} - 1 = |F|/|B| - 1 < n/|B|$.
Then, for all $\alpha \neq \alpha^*$, we have
\[ \inset{ p(\alpha) \suchthat p \in \mathcal{P}(\alpha^*) } \subset \inset{ \frac{\beta}{ \alpha - \alpha^* } \suchthat \beta \in B }, \]
and in particular this has dimension $1$ over $B$.  On the other hand, 
\[p(\alpha^*) = \left. \frac{\tr_{F/B}( \zeta(\bX - \alpha^*) )}{ \bX - \alpha^* } \right|_{\bX = \alpha^*} 
= \zeta. \]
Thus
\[ \inset{ p(\alpha^*) \suchthat p \in \mathcal{P}(\alpha^*) } = Z, \]
which is by definition full rank.  Thus, the conditions of Corollary~\ref{cor:polynomials} are satisfied, and the bandwidth of the resulting scheme is
\[ b = \sum_{\alpha \neq \alpha^*} \dim_B\inset{ p(\alpha) \suchthat p \in \mathcal{P}(\alpha^*) } = \sum_{\alpha \neq \alpha^*} 1 = n-1. \]
This proves Theorem~\ref{thm:fulllen}.
\end{proof}

\subsection{Large Field Sizes: Example Construction}\label{sec:bigt}
Theorem~\ref{thm:optimal} implies that our construction in Theorem~\ref{thm:fulllen} is optimal for Reed-Solomon codes. 
However, 
as we discussed in Section~\ref{sec:prelim}, the assumption that $A = F$ restricts $t \leq \log_2(n)$.  This is beneficial in some respects: if the extension field $F$ over $B$ has smaller degree, it is easier to implement in practice, especially when compared to constructions with $t = 2^n$.  However, when $t$ is large compared to $n-k$, this moves us to the regime where~\eqref{eq:bigt} is the binding lower bound, rather than~\eqref{eq:trivial}.  
In this large-$t$ regime, it's possible that the ratio of the bandwidth $b$ (the number of subsymbols downloaded) to $t$ (the number of subsymbols to be recovered) could tend to a constant.  On the other hand, in the small-$t$ regime where~\eqref{eq:trivial} is the binding constraint, this ratio must tend to infinity.  Thus, allowing $t$ to get large (increasing the level of sub-packetization) can improve the bandwidth in this sense.
We remark that while the trade-off between~\eqref{eq:bigt} and~\eqref{eq:trivial} occurs at $t = n-k$, it is conjectured in~\cite{twb14} that in fact $t$ must be exponentially large in $k$ in order for~\eqref{eq:bigt} to be attainable, and it's known that $t$ must be at least $k^2$~\cite{gtc}. 

With this trade-off in mind, we show how to use our framework to construct non-trivial linear repair schemes for RS codes with $A \subset F$ much smaller than the entire field.   This section is meant as a proof-of-concept; while our results are non-trivial, they are far from the bound~\eqref{eq:bigt}, and indeed the aformentioned result of~\cite{gtc} implies that with the particular parameters of our construction below, we cannot hope to attain that.

In a preliminary version of this work, we left it as an open question to take advantage of large $t$ with our approach.  In subsequent work~\cite{YB16} Ye and Barg were able to do this, and show how to meet~\eqref{eq:bigt} when $t$ is sufficiently large.  In that construction, $t = (n-k)^n$.
It is still an interesting open question how small one can take $t$ in order to come up with RS codes meeting~\eqref{eq:bigt}.

\begin{theorem}\label{thm:bigt}
Suppose that $F = GF(2^s)$ for even $s$.  Choose any even $n \leq 2( \sqrt{|F|} - 1)$.  
There is a set of $n$ evaluation points $A$ so that for any $k \leq n-2$, $RS(A,k)$
 admits a linear exact repair scheme over $B = GF(2^{s/2})$ so that the bandwidth in bits is at most
\[ b\log_2(|B|) \leq \frac{3}{4} sn. \]
In particular, choosing $k = n-2$, we have a extremely high-rate code with bandwidth
\[ b\log_2(|B|) \leq \inparen{\frac{3}{4} + o(1) } sk. \]
Notice that the naive scheme has bandwidth $sk$ bits, and \eqref{eq:bigt} gives a lower bound of $sk/2$ bits.
\end{theorem}
\begin{proof}
Let $\gamma$ be a primitive element of $F$.  Choose $A$ to consist of $n/2$ points from $B^*$ and $n/2$ points from $\gamma B^*$.
Choose the polynomials
\[ \mathcal{P}(\alpha^*) = \begin{cases} \inset{ 1, \bX } & \alpha^* \in \gamma B^* \\
\inset{ 1, \gamma^{-1}\bX } & \alpha^* \in B^* \end{cases} \]
It is easy to check that in either case, the set $\inset{p(\alpha^*) \suchthat p \in \mathcal{P}(\alpha^*) }$ has full rank, and that for all $\alpha \neq \alpha^*$, we have $\dim_B (\inset{ p(\alpha) \suchthat p \in \mathcal{P}(\alpha^*) } ) = 1$ whenever $\alpha \not\in \alpha^* B^*$.  
Finally, the polynomials in $\mathcal{P}(\alpha^*)$ are linear, and so by Corollary~\ref{cor:polynomials}, as long as $k \leq n-2$, this gives a linear exact repair scheme with bandwidth (in bits)
\[ \log(|B|) \inparen{ \inparen{ \frac{n}{2} - 2} + 2\inparen{\frac{n}{2} - 1} } = \frac{s}{2} \inparen{  \frac{3n}{2} - 1 }. \qedhere\]
\end{proof}

\subsection{A Specific Example: a $(14,10)$-GRS Code.}\label{sec:hdfs}
In this section, we give a linear exact repair scheme for the generalized Reed-Solomon code (see Definition~\ref{def:grs}) currently deployed in the Facebook Hadoop Analytics cluster. 
 Notice that for the exact repair problem, an exact repair scheme for
$RS(A, k)$ gives an exact repair scheme for $GRS(A, k, \lambda)$ for any $\lambda$.  Indeed, the $i$'th node holds the symbol $\lambda_i f(\alpha_i)$, and knows $\lambda_i$, so it also can compute $f(\alpha_i)$.

The HDFS-RAID~\cite{HDFS} module is an open-source module which implements coding for distributed storage in the Apache Hadoop Distributed File System (HDFS).
 This module allows for the use of Reed-Solomon code, and it implements a particular (14,10)-GRS code; this is currently deployed in the Facebook Hadoop Analytics cluster~\cite{xoringelephants}.
This GRS code was used as a benchmark in~\cite{xoringelephants} for comparison with novel regenerating storage schemes, and~\cite{scalarmds} improves on
the naive repair scheme for this GRS code (the naive scheme is the one implemented in the module). 
The latter work gives a non-trivial exact repair scheme for this particular code which can recover the systematic nodes (that is, the 10 out of the 14 nodes interpreted as holding the original data).

Using our characterization, it was quick to produce (via a computer search) a scheme that performs better than that of~\cite{scalarmds} and also which can recover all of the nodes, not only the systematic ones.  We give the details of the code, our search, and our results below.

The HDFS-RAID module (see~\cite{HDFS}, code at~\cite{github}, and the explicit generating matrix given in~\cite{scalarmds}) implements a GRS code over $F = GF(2^8)$ defined as follows.  Let $\zeta$ be a primitive element of $F$ (more precisely, $\zeta$ is a root of the primitive polynomial $1 + x^2 + x^3 + x^4 + x^8$).
The code is given by
\[ \cC = \inset{ (c_0,c_1,\ldots,c_{13}) \suchthat c(1) = c(\zeta) = c(\zeta^2) = c(\zeta^3) = 0 }, \]
where $c(\bX) = \sum_{i=0}^{13} c_i \bX^i$.
It is easiest to describe $\cC$ as above (using the dual formulation), but it is not hard to verify that $\mathcal{C}$ is indeed a GRS code: $\cC = GRS(A,k,\lambda)$ with evaluation points
\[ A = \inset{ 1,\zeta,\zeta^2,\ldots,\zeta^{13} } \]
and some vector $\lambda$.  As mentioned above, for the exact repair problem, only the evaluation points $A$ matter, and any exact repair scheme for $RS(A,k)$ will give an exact repair scheme for $\cC = GRS(A,k,\lambda)$.

The work~\cite{scalarmds} gives an improved scheme for $\mathcal{C}$.  More precisely, they show how to recover each of the $10$ systematic nodes using bandwidth at most $65$ bits; the naive bound is $80$ bits, and the lower bound~\eqref{eq:bigt} is $20$ bits.  We give a scheme which uses at most $64$ bits per node, and additionally can recover from any failure, not just the failure of a systematic node.  
We give a linear repair scheme over $B = GF(2^4)$.  Such a scheme is specified
by a choice of two degree-$3$ polynomials over $GF(2^8)$ for each $\alpha^* \in A$.
Our polynomials are given in Appendix~\ref{app:hdfs}, and our code can be found at \url{http://sites.google.com/site/marywootters/exhaust_FB.sage}. 

To find this scheme, we searched over all such polynomials which had three distinct roots in $A$; the reason for this assumption was to speed up the search, with the intuition that a value $p(\alpha) = 0$ for $\alpha \in A$ automatically reduces the dimension of the set $\inset{p(\alpha) \suchthat \alpha \in A}$.  This was a reasonably quick search and it produced a good solution to the exact repair problem for this particular code.  However, both the assumption about the roots and the large size of $B$ potentially limit the performance of this code; these assumptions were made so that naive search would be fast.  It is an interesting and important question if given evaluation points $A$ and a base field $B$, one can (sometimes) efficiently find a (near-)optimal linear exact recovery scheme for $RS(A,k)$ over $B$.

\section{Lower bounds}\label{sec:lower}
In this section, we prove Theorem~\ref{thm:optimal}, which gives a lower bound exactly matching Theorem~\ref{thm:fulllen}.
For the reader's convenience, we restate Theorem~\ref{thm:optimal} here.
\begin{theorem*}[Theorem~\ref{thm:optimal}]
Let $\mathcal{C}$ be an MDS code of dimension $k$ with evaluation points $A$ over a field $F$.  Let $B \leq F$ be a subfield.  Any linear repair scheme for $\cC$ over $B$ must have bandwidth (measured in subsymbols of $B$) at least
\[ b \geq (n-1) \log_{|B|} \inparen{ \frac{n-1}{n-k} }. \]
In particular, the bandwidth (measured in bits) for any linear repair scheme over any base field $B$ is at least
\[ b\log_2(|B|) \geq (n-1) \log_2\inparen{ \frac{n-1}{n-k} }. \]
\end{theorem*}

\begin{proof}[Proof of Theorem~\ref{thm:optimal}] 
Fix any $\alpha^* \in A$, and 
consider any linear exact repair scheme which repairs the node corresponding to $\alpha^*$ using $b$ sub-symbols from $B$.
By Theorem~\ref{thm:polynomials}, there is some set $\mathcal{P} \subset \cC^\perp$ of size $t$ so that 
\begin{equation}\label{eq:fullrank}
 \inset{ p(\alpha^*) \suchthat p \in \mathcal{P} }  
\end{equation}
has full rank over $B$, and so that
\begin{equation}\label{eq:smallrank}
 \dim_B (\inset{ p(\alpha) \suchthat p \in \mathcal{P} }) = d_\alpha
\end{equation}
where
\[ b = \sum_{\alpha \in A \setminus \inset{\alpha^*}} d_\alpha. \]
For any vector $x \in B^t$ (indexed by the functions $p \in \mathcal{P}$) let $p_x: F \to F$ be 
\[ p_x(\bX) = \sum_{p \in \mathcal{P}} x_p \cdot p(\bX). \]
Let
\[ S_\alpha = \inset{ x \in B^t \suchthat p_x(\alpha) = 0 }. \]
By \eqref{eq:smallrank}, $S_\alpha$ is a vector space over $B$ of dimension $t - d_\alpha$.
Thus, on average over all nonzero $x \in B^t$, we have
\begin{align*}
 \frac{1}{|B|^t - 1}\sum_{x\neq0,x \in B^t} |\inset{ \alpha \in A \setminus \inset{\alpha^*} \suchthat x \in S_\alpha }| &= \frac{1}{|B|^t - 1} \sum_{\alpha \in A \setminus \inset{\alpha^*}} |B|^{t - d_\alpha} \\
&= \inparen{ \frac{ |F| }{|F| - 1} }\sum_{\alpha \in A \setminus \inset{\alpha^*}} |B|^{-d_\alpha} \\
&=: r.
\end{align*}
In particular, there exists some $x^* \in B^t$ so that
$ |\inset{\alpha \suchthat x^* \in S_\alpha}| \geq r. $
Consider 
\[ p^*(\bX) = \sum_{p \in \mathcal{P}} x^*_p p(\bX). \]
By the choice of $x^*$, $p^*$ vanishes on at least $r$ points of $A \setminus \inset{ \alpha^*}$.  
Notice also that $p^*$ is nonzero.  Indeed, if it were zero, then 
\[ \sum_{p \in \mathcal{P}} x^*_p p(\alpha^*) = 0, \]
contradicting~\eqref{eq:fullrank}.
However, $p^* \in \cC^\perp$, since $\cC^\perp$ is a linear code.  Since $\cC$ is an MDS code, so is $\cC^\perp$, and in particular, the distance is $n-k$.  Since $p^*$ is nonzero, this implies that $r < n-k$, which implies that
\[\sum_{\alpha \in A \setminus \inset{\alpha^*}} |B|^{-d_\alpha} < n-k. \]
Thus, we have the bound on bandwidth
\[ b \geq \min_{d_\alpha \in [0,t]} \sum_{\alpha \in A \setminus \inset{\alpha^*}} d_\alpha  \qquad s.t. \qquad \sum_{\alpha \in A \setminus \inset{\alpha^*}} |B|^{-d_\alpha} < n-k. \]
The minimum occurs when are $d_\alpha$ are balanced, and equal to $\log_{|B|}\inparen{ \frac{ n-1}{n-k}}$, and we have
\[ b \geq (n-1) \log_{|B|} \inparen{ \frac{n-1}{n-k} }. \]
This completes the proof.
\end{proof}

\begin{remark}[A simple proof of the cut-set bound for exact repair of MDS codes]
The bound~\eqref{eq:bigt}, which holds for general codes and for functional repair, has a very simple proof for linear exact repair of MDS codes.  In the language of the proof of Theorem~\ref{thm:optimal}, the number of roots of $\mathcal{P}_x(\bX)$ must be less than $n-k$ for all nonzero $x \in B^t$.  This implies that for all sets $T \subset A \setminus \inset{\alpha^*}$ of size $n-k$, we must have $\sum_{\alpha \in T} d_\alpha \geq t$, or else there would be some $x$ so that $\mathcal{P}_x$ vanishes on $T$.  Thus, averaging over all sets $T$, we obtain
\[ b = \sum_{\alpha \in A \setminus \inset{\alpha^*}} d_\alpha \geq t \inparen{ \frac{n-1}{n-k} }, \]
which is precisely \eqref{eq:bigt}.
\end{remark}

\section{Discussion}\label{sec:conclusion}

Inspired by the exact repair problem for Reed-Solomon codes, we studied a variant of the classical polynomial interpolation problem.  How many bits are needed from $\inset{ f(\alpha) \suchthat \alpha \neq \alpha^*}$ in order to recover $f(\alpha^*)$?  We have shown that this can be much smaller than the number of bits needed using standard polynomial interpolation.  We gave a characterization of the number of bits needed, in terms of another problem about polynomials.

Our results imply that, while Reed-Solomon codes are often given as an example of how \em not \em to solve the exact repair problem, in fact, with the right reconstruction algorithm, they can be optimal!  This is heartening news, since RS codes are commonly used in practice.  
More precisely, we give an optimal linear exact repair scheme when block length is $n = |F|$; to the best of our knowledge, for this level of sub-packetization, Reed-Solomon codes significantly out-perform all known schemes.   
Additionally, we give a few examples of how to use this characterization in order to come up with non-trivial repair schemes for other codes.  In particular, we give an example of a family where the set of evaluation points $A \subset F$ is much smaller than the entire field.  We also give an improved exact repair scheme for a particular $(14,10)$-GRS code used by Facebook.  

Finally, subsequent work using our framework has shown that RS codes can also perform well in the ``large-$t$"-regime, and in fact can approach the cut-set bound~\cite{YB16}.  

We conclude with a few open questions.
\begin{enumerate}
\item Given a specific RS code (that is, given a set $A \subset F$), is there an efficient algorithm which will give a (near)-optimal linear repair scheme?  Our example with the $(14,10)$-GRS code was small enough that we could do an exhaustive search, under some additional assumptions, but even for this code we still do not know the best linear repair scheme.  
\item Our scheme when $A = F$ is very efficient in terms of bandwidth, but is not very efficient in terms of the total number of bits accessed.  An inspection of the scheme reveals that most nodes will have to touch $\Omega(t)$ bits before deciding which $O(1)$ bits to return.  In practice, this is also an important concern.
How well Reed-Solomon codes can do when this is taken into account.
\item Our characterization and examples are for \em linear \em repair schemes.  How much better can one do with non-linear repair schemes?  
\end{enumerate}

\section*{Acknowledgements}
We thank Alex Dimakis for bringing this problem to our attention and for patiently answering our questions about the problem set-up, and Ankit Rawat for helpful feedback.  We also thank the Simons Institute at Berkeley for their hospitality in Spring 2015, where some of this research was done.

\bibliographystyle{alpha}
\bibliography{refs}

\appendix
\section{Table of notation}\label{app:notation}
For reference,
Table~\ref{fig:translation} gives a summary of our notation, and also a translation to common conventions from the regenerating codes literature.
Broadly speaking, we reserve some Greek letters ($\alpha, \beta$,...) for elements of the finite field $F$ and a few $(\eps, \delta)$ for small real numbers; we reserve capital Roman letters ($S, T, A$,...) for subsets of the finite field $F$; and we reserve some lower-case Roman letters ($b,t,d$,...) for integers and some $(f, p)$ for polynomials mapping $F \to F$.  We use $\bX$ as a variable in polynomials.

\begin{figure}[ht!]
\begin{tabular}{|>{\centering\arraybackslash}p{2in}|>{\centering\arraybackslash}p{1in}|>{\centering\arraybackslash}p{1.5in}|>{\centering\arraybackslash}p{1.5in}|}
\hline
Description in terms of the exact repair problem & Our name & Our notation & Common notation in regenerating codes \\
\hline
Contents of a node & Symbol & An element of $F$ & -- \\
\hline
Response from a node & Sub-symbol (or several sub-symbols) & Element(s) of $B \leq F$ & Smallest unit of subpacketization; often $\F_q$ or $\F_p$\\
\hline
Number of blocks in the original file & Message length & $k$ & $k$ \\
\hline
Number of nodes & Codeword length & $n$ & $n$ \\
\hline
Number of sub-symbols downloaded from each node & -- & We allow this to vary from node to node in our definition
& $\beta$ \\
\hline
Number of sub-symbols contained in each node & Symbol size & $t$ & $\alpha$ \\
\hline
Number of sub-symbols in the original file & & $kt$ & $\mathcal{M}$ \\
\hline
Number of sub-symbols contained in each message block & -- & $t$ & $\mathcal{M}/k$ \\
\hline
Number of sub-symbols downloaded to repair a single erased node & (Exact) repair bandwidth & $b$ & $\gamma$\\
\hline
Number of nodes accessed to repair a single erased node & (Exact) repair locality & $d$ & $d$ \\
\hline
\end{tabular}
\caption{A description of our notation and a translation to the standard notation for regenerating codes.  We consider minimum storage regenerating (MSR) codes, so the number of sub-symbols stored in a single node is the same as the number of subsymbols stored in a block of the file.}
\label{fig:translation}
\end{figure}

\section{Direction of communication}\label{app:justification}

It may seem strange that the nodes are allowed to return \em any \em (linear) function of the data stored in them.  After all, the goal is to minimize communication, and this scheme requires that we ask each remaining node for a specific (set of) functions.  This could potentially result in a lot of communication from the new replacement node to the existing nodes; meanwhile, the definition of repair bandwidth only captures communication from the existing nodes to the replacement node.  
However, while minimizing communication in both directions is obviously of interest, the regenerating codes literature has focused mostly on the communication from the existing nodes to the newcomer.    The main justification is as follows.

First, notice that we only need to communicate the identity of the failed node, $\alpha^*$, to each remaining node.  If $t$ is the dimension of $F$ over $B$, this requires $t$ sub-symbols to be sent to each remaining node.  Thus, if either (a) not too many nodes need to be contacted or (b) it is easy to broadcast the information $\alpha^*$ to many nodes, this is not an issue.   

Moreover, even if the number of nodes contacted is large compared to the dimension $t$ of $F$ over $B$, this is not a problem in practice.  Because each server is very large, and these codes are typically implemented over smaller fields, the following set-up---called ``data striping", or an ``interleaved RS code"---is common.\footnote{See~\cite{SRKR09} or \cite{RSK11} for a further description of this; we also thank Alex Dimakis for pointing out to us that this is how Reed-Solomon codes are used by Facebook in~\cite{github}.} 
We still have $n$ servers, each associated with an element $\alpha \in A$; however, instead of holding a single element of $F$, these servers hold $m$ elements of $F$.
We encode a
file consisting of $m k \log_2(|F|)$ bits as follows.  First, we break up the
file into $m$ messages in $F^k$, interpreted as $m$ polynomials $f^{(1)}, f^{(2)}, \ldots, f^{(m)}$ of degree at most $k-1$.  
Then we encode each of these messages with the Reed-Solomon code, obtaining $m$ codewords $(f^{(i)}(\alpha))_{\alpha \in A}$.
Finally, we distribute these codewords among the $n$ servers: the server corresponding to $\alpha$ holds $f^{(i)}(\alpha)$ for $i=1,\ldots,m$.
This setup is depicted in Figure~\ref{fig:practice}.

Now, suppose the server corresponding to $\alpha^*$ is erased.  We wish to set up a replacement server, and we run our exact repair scheme for each codeword $c^{(i)}$, $i=1,\ldots,m$.   Notice that the replacement server needs to regenerate at least $m$ elements of $F$, so the number of sub-symbols it must download from the remaining servers is at least $mt$.  This swamps the amount of communication going in the other direction: in this set-up, the replacement server needs to communicate $\alpha^*$ to each of the $n$ servers, which is $nt \ll mt$ sub-symbols.  

Thus, even in the case where $t$ is small compared to $n$, it is interesting to consider the one-way communication captured in the definition of the repair bandwidth. 

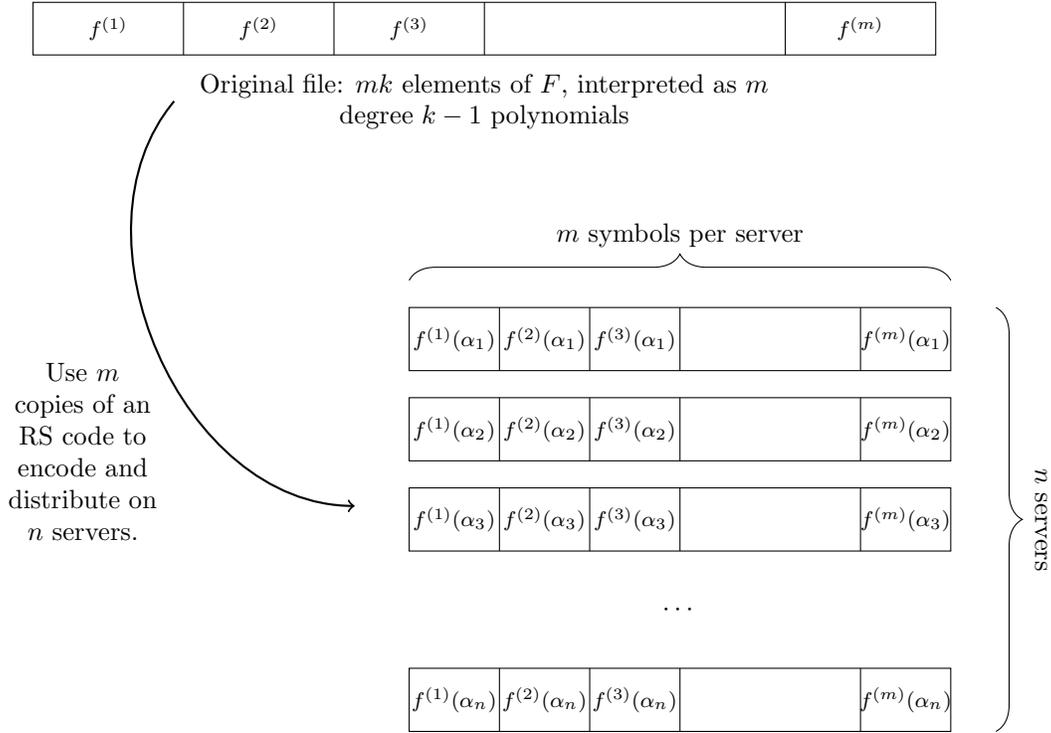
\begin{figure}
\begin{center}
\begin{tikzpicture}
\def\nn{5}
\def\mm{6}
\begin{scope}[xshift=-5cm]
\draw (0,0) rectangle (\mm*2,  .7 );
\node(target) at (\mm, 0) {};
\foreach \j in {1,...,3}
{
	\draw (2*\j,0) -- (2*\j,  .7);
	\node at (2*\j - 1, .35) {\footnotesize$f^{(\j)}$};
}
\draw (2*\mm-2,0) -- (2*\mm-2,0.7);
\node at (2*\mm - 1, .35) {\footnotesize$f^{(m)}$};
\node[below=0cm of target](anchor)  {\begin{minipage}{8cm} \begin{center}{Original file: $mk$ elements of $F$, interpreted as $m$ degree $k-1$ polynomials}\end{center}\end{minipage}};
\end{scope}
\begin{scope}[yshift=-9cm,scale=1.2]
\node(pt) at (-.5, \nn/2) {};
\foreach \i in {1,2,3}
{
	\draw (0,\nn-\i) rectangle (\mm, \nn-\i + .7 );
	\foreach \j in {1,...,3}
	{
		\draw (\j,\nn-\i) -- (\j,\nn-\i + .7);
		\node at (\j - 0.5,\nn-\i+.35) {\footnotesize$f^{(\j)}(\alpha_{\i})$};
	}
	\foreach \j in {\mm}
	{
		\draw (\j-1,\nn-\i) -- (\j-1,\nn-\i + .7);
		\node at (\j - 0.5, \nn-\i+.35) {\footnotesize$f^{(m)}(\alpha_{\i})$};
	}
}
\node at (\mm/2, \nn - 4 + .35) {$\cdots$};
\foreach \i in {\nn}
{
	\draw (0,\nn-\i) rectangle (\mm, \nn-\i + .7 );
	\foreach \j in {1,...,3}
	{
		\draw (\j,\nn-\i) -- (\j,\nn-\i + .7);
		\node at (\j - 0.5, \nn-\i+.35) {\footnotesize$f^{(\j)}(\alpha_{n})$};
	}
	\foreach \j in {\mm}
	{
		\draw (\j-1,\nn-\i) -- (\j-1,\nn-\i + .7);
		\node at (\j - 0.5, \nn-\i+.35) {\footnotesize$f^{(m)}(\alpha_{n})$};
	}
}
\draw [decorate,decoration={brace,amplitude=10pt},xshift=0pt,yshift=0pt]
(\mm + .5, \nn - .3) -- (\mm + .5, 0) node [black,midway,xshift=0.6cm, rotate=-90] 
{$n$ servers};
\draw [decorate,decoration={brace,amplitude=10pt},xshift=0pt,yshift=0pt]
(0, \nn) -- (\mm , \nn) node [black,midway,yshift=0.6cm] 
{$m$ symbols per server};
\end{scope}
\draw[thick,->] (anchor.west) to[out=230,in=180] node[midway,below,xshift=-1cm]{\begin{minipage}{2cm}\begin{center}Use $m$ copies of an RS code to encode and distribute on $n$ servers.\end{center}\end{minipage}} (pt);
\end{tikzpicture}
\end{center}
\caption{How we might implement an RS code for storage in practice, when the size of a server is large compared to the size of an element of $F$. 
Suppose we have an exact repair scheme for the RS codes.
When a server corresponding to $\alpha^*$ fails, the replacement server communicates ``$\alpha^*$" to each surviving server.  Then each surviving server sends the replacement server what the exact repair scheme dictates for each of the $m$ RS codewords.
In this set-up, the amount of communication from the replacement server to the existing servers (which is at most $nt$) is negligible compared to the amount of communication (which is $mb$) from the existing nodes to the replacement node, since $m$ is much larger than all other parameters.}
\label{fig:practice}
\end{figure}

\section{Extended literature review}\label{app:litreview}
In Table~\ref{fig:exactmsr}, we give an extended literature review summarizing work on MDS codes for the exact repair problem, and including our results for RS codes in comparison.

\begin{figure}
\footnotesize
\begin{center} Relevant parameter regime: $t \leq n-k$ \end{center}
\begin{tabular}{|>{\centering\arraybackslash}p{.7in}|>{\centering\arraybackslash}p{.65in}|>{\centering\arraybackslash}p{.55in}|>{\centering\arraybackslash}p{.8in}|>{\centering\arraybackslash}p{.8in}|>{\centering\arraybackslash}p{1in}|>{\centering\arraybackslash}p{.8in}|}
\hline
Paper &  $k,n$ & No. nodes accessed $d$ & No. Subsymbols per symbol $t$ & Subsymbol size $|B|$ & Repair bandwidth $b$ & Notes \\
\hline\hline
\rowcolor{yellow!20}
Naive lower bound & -- & -- & -- & -- & $b \geq k + t -1$ & \\
\hline
\cite{WD09} & $n \geq k+2$ & $d = k + t - 1$ & $t \leq n-k$ & $|B|$ may be as large as $n^k$ & $b = (k-1)t + 1$ & \\
\hline
\cite{SRKR09} & $n \geq k + 2$  & $d = k + t - 1$ & $t \geq k - 2$ & $|B| \geq t + n - k $ & $b = k + t - 1$ & For exact repair of systematic nodes only.\\
\hline
\hline
\rowcolor{yellow!20}
Theorem~\ref{thm:optimal} (holds for RS codes with any evalation points).
& $k = (1 -\eps)n$ & -- & -- & -- & $b \geq$ \newline $(n-1)\log_2(\frac{1}{\eps}) - O(1) $ & Lower bound for linear repair schemes for RS codes. \\
\hline
Theorem~\ref{thm:fulllen} (RS codes with $A=F$) & $k = (1-1/|B|)n$  & $d = n-1$ & $t = \log_{|B|}(n)$ & $B \leq F$ subfield & $b = (n-1)$& \\
\hline
Corollary~\ref{cor:binary} (RS codes with $A=F$) & $k = (1-\eps)n$  & $d = n-1$ & $t = \log_{2}(n)$ & $|B| = 2$  & $b = (n-1)\log_2(1/\eps)$& Provided $\log_2(1/\eps) | t$ \\
\hline
\end{tabular}

\begin{center} Relevant parameter regime: $t \geq n-k$ \end{center}
\begin{tabular}{|>{\centering\arraybackslash}p{.7in}|>{\centering\arraybackslash}p{.65in}|>{\centering\arraybackslash}p{.55in}|>{\centering\arraybackslash}p{.8in}|>{\centering\arraybackslash}p{.8in}|>{\centering\arraybackslash}p{1in}|>{\centering\arraybackslash}p{.8in}|}
\hline
Paper &  $k,n$ & No. nodes accessed $d$ & No. Subsymbols per symbol $t$ & Subsymbol size $|B|$ & Repair bandwidth $b$ & Notes \\
\hline\hline
\rowcolor{yellow!20}
\cite{DGWW10}  & -- & -- & -- & -- & $b \geq \frac{ t d }{d + 1-k}$ & Lower bound for functional repair\\
\hline
\cite{SR10}  & $n \leq 2k$  & $d =n-1 \geq 2k - 1$ & $t = d-k+1 \geq k$ & $|B| = n-k$ & $b = \frac{ t d }{d + 1 - k}$ & \\
\hline
\cite{RSK11} & $n \leq 2k + 1$  & $d \geq 2k - 2$ & $t = d - k + 1 \geq k-1 $ & $|B| = n-k$ & $b = \frac{td}{d + 1 - k}$ & \\
\hline
\cite{SR10b, CJM13} & -- & -- & $(d+1-k) \cdot \Delta^{c}$ & $|B|$ depends on $\Delta$ & $b \leq d(\Delta + 1)^{c}$ (which is about $\frac{ dt}{d + 1 - k}$ for large $\Delta$)  & $c = c(n,k) = (n-k)(k-1)$\\
\hline\hline
Theorem~\ref{thm:bigt} (RS codes with a specific $A \subset F$) & $k \leq n-2$ & $d = n-1$ & Any $t$& $|B| = 2$ & $b = kt \inparen{\frac{3}{4} - o(1) } $ & The naive result is $kt$ and the lower bound \eqref{eq:bigt} is $kt/2$, so this is non-trivial but not optimal. \\
\hline
\end{tabular}
\normalsize
\caption{\textbf{Constructions and bounds for exact repair of general MDS codes.}
Above we summarize results for the exact repair problem for MDS codes, separated by parameter regime.  The yellow rows indicate lower bounds, and the white rows are upper bounds.
We note that there are several results known in the large-$t$ setting ($t$ exponential in $n$) for exact recovery of systematic nodes only and for specific rates, like $n = k+2$~\cite{CHL11,TWB13,PDC13}.  We omit the quantitative details of these results in our table, and refer the reader to~\cite{wiki} for a more extensive literature review in these special cases.
}\label{fig:exactmsr}
\end{figure}

\section{Explicit construction for Facebook code}\label{app:hdfs}
In this section, we present the polynomials returned from our search, which give a linear exact recovery scheme for the $(14,10)$-GRS code used in the Facebook Hadoop Analytics cluster.  The code which produced these polynomials can be found at \url{http://sites.google.com/site/marywootters/exhaust_FB.sage}.
\begin{figure}[h]
\begin{center}
\begin{tabular}{|c|cc|c|}
\hline
$ \alpha^*$   & Polynomials & &Bandwidth (in bits) for $ \alpha^*$   \\ 
\hline$ \zeta^{0}$ & $  (\bX + \zeta^{1})(\bX + \zeta^{2})(\bX + \zeta^{5})$ & $  (\bX + \zeta^{3})(\bX + \zeta^{8})(\bX + \zeta^{6})$ & 64\\ \hline 
$ \zeta^{1}$ & $  (\bX + \zeta^{2})(\bX + \zeta^{3})(\bX + \zeta^{6})$ & $  (\bX + \zeta^{4})(\bX + \zeta^{9})(\bX + \zeta^{7})$ & 64\\ \hline 
$ \zeta^{2}$ & $  (\bX + \zeta^{3})(\bX + \zeta^{9})(\bX + \zeta^{6})$ & $  (\bX + \zeta^{3})(\bX + \zeta^{13})(\bX + \zeta^{12})$ & 60\\ \hline 
$ \zeta^{3}$ & $  (\bX + \zeta^{2})(\bX + \zeta^{9})(\bX + \zeta^{6})$ & $  (\bX + \zeta^{2})(\bX + \zeta^{13})(\bX + \zeta^{12})$ & 60\\ \hline 
$ \zeta^{4}$ & $  (\bX + \zeta^{2})(\bX + \zeta^{9})(\bX + \zeta^{6})$ & $  (\bX + \zeta^{2})(\bX + \zeta^{13})(\bX + \zeta^{12})$ & 60\\ \hline 
$ \zeta^{5}$ & $  (\bX + \zeta^{1})(\bX + \zeta^{3})(\bX + \zeta^{9})$ & $  (\bX + \zeta^{3})(\bX + \zeta^{4})(\bX + \zeta^{11})$ & 64\\ \hline 
$ \zeta^{6}$ & $  (\bX + \zeta^{1})(\bX + \zeta^{2})(\bX + \zeta^{10})$ & $  (\bX + \zeta^{1})(\bX + \zeta^{5})(\bX + \zeta^{12})$ & 64\\ \hline 
$ \zeta^{7}$ & $  (\bX + \zeta^{1})(\bX + \zeta^{2})(\bX + \zeta^{8})$ & $  (\bX + \zeta^{1})(\bX + \zeta^{6})(\bX + \zeta^{12})$ & 64\\ \hline 
$ \zeta^{8}$ & $  (\bX + \zeta^{2})(\bX + \zeta^{9})(\bX + \zeta^{6})$ & $  (\bX + \zeta^{2})(\bX + \zeta^{13})(\bX + \zeta^{12})$ & 60\\ \hline 
$ \zeta^{9}$ & $  (\bX + \zeta^{1})(\bX + \zeta^{2})(\bX + \zeta^{5})$ & $  (\bX + \zeta^{3})(\bX + \zeta^{8})(\bX + \zeta^{6})$ & 64\\ \hline 
$ \zeta^{10}$ & $  (\bX + \zeta^{1})(\bX + \zeta^{2})(\bX + \zeta^{5})$ & $  (\bX + \zeta^{1})(\bX + \zeta^{6})(\bX + \zeta^{13})$ & 64\\ \hline 
$ \zeta^{11}$ & $  (\bX + \zeta^{2})(\bX + \zeta^{9})(\bX + \zeta^{6})$ & $  (\bX + \zeta^{2})(\bX + \zeta^{13})(\bX + \zeta^{12})$ & 60\\ \hline 
$ \zeta^{12}$ & $  (\bX + \zeta^{1})(\bX + \zeta^{2})(\bX + \zeta^{5})$ & $  (\bX + \zeta^{1})(\bX + \zeta^{6})(\bX + \zeta^{13})$ & 64\\ \hline 
$ \zeta^{13}$ & $  (\bX + \zeta^{1})(\bX + \zeta^{2})(\bX + \zeta^{5})$ & $  (\bX + \zeta^{3})(\bX + \zeta^{8})(\bX + \zeta^{6})$ & 64\\ \hline 
\end{tabular}
\end{center}
\caption{Polynomials which give an exact repair scheme for the $(14,10)$-GRS code used in~\cite{HDFS}.}\label{fig:hdfstable}
\end{figure}
Table~\ref{fig:hdfstable} gives an exact repair scheme for $RS(A,10)$ with $A = \inset{ \zeta^i : 0 \leq i \leq 13 } \subset GF(2^8)$ over $B = GF(2^4)$.  As in Theorem~\ref{thm:polynomials}, such a scheme is given by two cubic polynomials for each choice of $\alpha^*$.  We have also listed the bandwidth
\[ b(\alpha^*) \cdot \log_2(|B|) = 4 \cdot \sum_{\alpha \in A \setminus \inset{\alpha^*} } \dim_B \inset{ p(\alpha) \suchthat p \in \mathcal{P}(\alpha^*) } \]
for each choice of $\alpha^*$.  The total bandwidth for the scheme is given by the maximum of these which is $64$ bits.  As discussed in Section~\ref{sec:hdfs}, an exact repair scheme for $RS(A,10)$ gives an exact repair scheme for the code $GRS(A,10,\lambda)$ used in~\cite{xoringelephants,scalarmds,HDFS}.
The previous best result for this code, from~\cite{scalarmds} had bandwidth $65$ bits and only recovered the $10$ systematic nodes.

\end{document}